\documentclass[sigconf,nonacm]{acmart}

\usepackage{amsmath}
\usepackage{graphicx}
\usepackage{acronym}
\usepackage{siunitx}
\usepackage{subcaption}
\usepackage{amsthm}
\newtheorem{lemma}{Lemma}
\newtheorem{theorem}{Theorem}

\acrodef{MC}{molecular communication}
\acrodef{CIR}{channel impulse response}
\acrodef{PBS}{particle-based simulation}
\acrodef{1D}{one-dimensional}
\acrodef {3D}{three-dimensional}
\acrodef{PDF}{probability density}
\acrodef{CVS}{cardiovascular system}
\acrodef{CFD}{computational fluid dynamics}

\acmYear{2026}
\copyrightyear{2026}

\title{Analytical Modeling of Dispersive Closed-loop MC Channels with Pulsatile Flow}

\author{Theofilos Symeonidis\textsuperscript{\footnotesize 1},
Fardad Vakilipoor\textsuperscript{\footnotesize 1},
Robert Schober\textsuperscript{\footnotesize 1},
Nunzio Tuccitto\textsuperscript{\footnotesize 2},
Maximilian Schäfer\textsuperscript{\footnotesize 1}}

\affiliation{%
  \institution{\textsuperscript{1}Friedrich-Alexander-Universität Erlangen-Nürnberg (FAU), Erlangen, Germany\\
  \textsuperscript{2}University of Catania, Catania, Italy}
  \city{}\country{}
}

\renewcommand{\shortauthors}{Symeonidis et al.}

\begin{document}

\begin{abstract}
\Ac{MC} is a communication paradigm in which information is conveyed through the controlled release, propagation, and reception of molecules. Many envisioned healthcare applications of \ac{MC} are expected to operate inside the human body. In this environment, the \ac{CVS} acts as the physical channel, which forms a closed-loop network where particle transport is mainly governed by the combined effects of diffusion and flow. Despite the fact that physiological flows in many parts of the human body are inherently pulsatile due to the cardiac cycle, most existing models for dispersive closed-loop \ac{MC} channels assume a constant flow velocity. In this paper, we present a time-variant \ac{1D} channel model for dispersive closed-loop \ac{MC} systems with pulsatile flow. We derive an analytical expression for the \ac{CIR}, which follows a wrapped Normal distribution with time-variant mean and variance. The obtained model reveals the cyclostationary nature of the channel and quantifies the influence of pulsation on the temporal concentration profile compared to steady-flow systems. Finally, the model is validated by \ac{3D} \acp{PBS}, showing excellent agreement and enabling an efficient analytical characterization of the channel.
\end{abstract}




\maketitle

\acresetall

\section{Introduction}
\Ac{MC} is an emerging interdisciplinary research field that facilitates information transfer in biological environments. MC has a wide range of potential biomedical and healthcare applications, including health monitoring, early disease detection, and targeted treatment~\cite{6548006,8710366,8567997}. 
Most of these applications are expected to operate inside the human body, specifically in the \ac{CVS}, which constitutes the physical propagation channel, where signaling molecules are transported assisted by blood flow~\cite{FELICETTI201627,10433775}. Therefore, it is crucial to derive meaningful models that account for the physical properties and behavior of the channel.

Many existing studies on flow-based \ac{MC} consider molecule transport in simplified vascular environments. In these studies, the propagation of signaling molecules is typically modeled as an analytical advection-diffusion process in a single straight pipe~\cite{8647632,9359669,6548006}. Larger vascular structures have been represented by equivalent network models that capture branching and flow distribution~\cite{jakumeit2026mixtureinversegaussianshemodynamic,jakumeit2025molecularsignalreceptioncomplex}, but still rely on non-recirculating flow paths. Furthermore, several experimental platforms have been developed to investigate flow-based molecular transport under controlled conditions~\cite{Furubayashi2017DesignAW,8924625}.

However, in practical in-body \ac{MC} systems, the carrier fluid is not confined to a non-recirculating propagation path but instead moves within a closed-loop network due to the continuous circulation of blood. 
This recirculation fundamentally alters the channel characteristics, introducing long-term memory, repeated molecule arrivals, and strong coupling between local propagation dynamics and the network's overall structure~\cite{brand2025closedloop}. Therefore, suitable channel models have to explicitly account for the closed-loop nature of the \ac{CVS} in order to accurately describe the transport of signaling molecules in realistic physiological environments.

Recent studies have started to investigate \ac{MC} in closed-loop systems~\cite{ieee8618321,brand2025closedloop}. These studies demonstrate that circulation leads to very different channel characteristics compared to non-recirculating systems, including extended delay spreads, persistent inter-symbol interference, and a dependence of the received signal on the topology. Experimental studies further confirm the presence of these characteristics, showing that molecule transport in closed-loop systems exhibits distinctive temporal patterns that cannot be replicated by open-loop models~\cite{vakilipoor2025cam,BrandICC2024ClosedLoop}. These observations underline the importance of accounting for the closed-loop nature of the channel in the development of suitable channel models.

Most existing studies on closed-loop \ac{MC} systems are based on the assumption of steady flow conditions~\cite{ieee8618321,brand2025closedloop,vakilipoor2025cam,BrandICC2024ClosedLoop}, despite the fact that blood flow in the \ac{CVS} is inherently pulsatile, due to the cardiac cycle. This physiological feature is well established in the hemodynamics and fluid mechanics literature~\cite{https://doi.org/10.1002/j.2040-4603.2016.tb00700.x,https://doi.org/10.1002/smll.201904032} and is further supported by in-vivo measurements and numerical studies that provide realistic time-variant blood flow shapes~\cite{https://doi.org/10.1002/fld.1966,Holdsworth_1999}.

\begin{figure*}[t]
  \centering
  \includegraphics[width=0.9\textwidth]{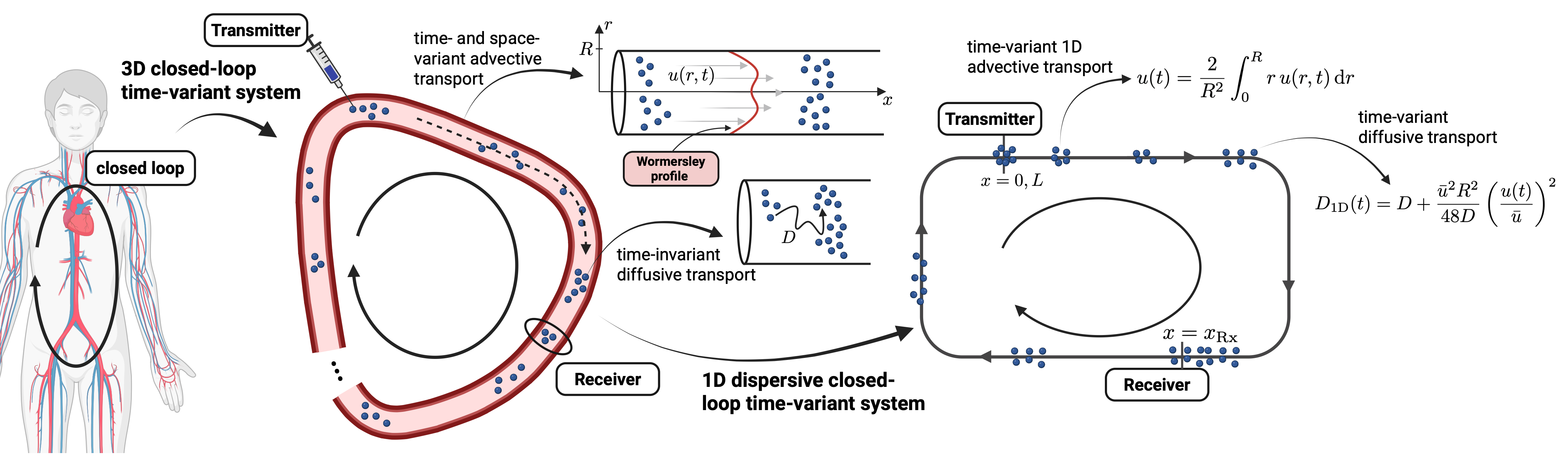}
  \vspace*{-0.5cm}
  \caption{\textbf{Closed-loop molecular communication system and model reduction.}
Signaling molecules are released by the transmitter and propagate through a vascular-like closed-loop channel under pulsatile flow (left). The resulting transport is governed by a time- and space-variant advection-diffusion process with Womersley velocity profile (center). In the dispersive regime, the system can be approximated by a one-dimensional model with time-variant effective velocity and diffusion coefficient (right).}
  \label{fig:system_model}
\end{figure*}

A number of \ac{MC} studies have adopted pulsatile flow models and demonstrated their influence on molecule transport and channel responses. In~\cite{6548006}, an analytical model for \ac{MC} in the cardiovascular system with pulsatile blood flow is developed, characterizing the time-variant delivery rate at the receiver. However, the resulting formulation is highly complex and does not yield a tractable analytical \ac{CIR}, limiting its applicability for communication-theoretic analysis. In~\cite{11157771}, nanoparticle propagation with pulsatile flow is investigated using numerical particle simulations, revealing significant deviations of the \ac{CIR} compared to the steady-flow case. However, the study does not provide an analytical channel model, neglects diffusive transport, and is tailored to a specific physiological scenario, limiting its general applicability. In~\cite{11161338}, analytical \acp{CIR} for non-Newtonian flow in a plaque-obstructed vessel are derived and complemented by detailed \ac{CFD} simulations incorporating pulsatile flow effects. However, the study relies heavily on simulations, while the analytical models are shown to provide only rough approximations of the actual flow behavior and do not yield a unified analytical description of pulsatile channels. Furthermore, the considered scenario is highly complex and application-specific. In~\cite{lee2026correctedinversegaussianfirsthittingtimemodelingmolecular}, an analytical channel model for first-hitting-time under time-variant drift is derived. However, the model is restricted to one-dimensional propagation and focuses on first-hitting-time statistics for absorbing receivers, rather than general \acp{CIR}. Moreover, pulsatile flow is represented via a spatially uniform time-variant drift and approximations are used to obtain a tractable expression, limiting its applicability to realistic flow conditions.
In addition, none of the above works considers closed-loop propagation. 

In fact, the combined impact of flow recirculation and pulsatile flow in the \ac{CVS} is not considered in any existing \ac{MC} channel models.

Motivated by these observations, we propose a closed-loop \ac{MC} channel model that explicitly accounts for pulsatile flow, thereby capturing the combined impact of recirculation and time-variant transport dynamics. It builds on recent results in~\cite{wang2025pulsatile} showing that, in the dispersive regime~\cite{Jamali2019ChannelModeling}, pulsatile flow with a Womersley velocity ~\cite{womersley1955} can be approximated by a \ac{1D} advection-diffusion equation with time-variant velocity. 
Based on this reduced description, we derive an analytical expression for the resulting \ac{CIR}. To the best of our knowledge, this is the first analytical \ac{CIR} representation for dispersive closed-loop \ac{MC} with pulsatile flow. The main contributions of this paper are as follows:

\begin{itemize}
\item We extend an existing closed-loop \ac{MC} channel model~\cite{vakilipoor2025cam} by incorporating a physiologically motivated time-variant flow velocity, which enables the modeling of molecule transport under pulsatile flow in closed-loop environments.

\item Based on this model, we derive an analytical expression for the \ac{CIR} of dispersive closed-loop channels with pulsatile flow, which takes the form of a wrapped Normal distribution with time-variant mean and variance.

\item We validate the proposed analytical model by comparison to the results from \ac{3D} \acp{PBS} of the underlying advection-diffusion process, and study how pulsatile flow affects the \ac{CIR} compared to steady-flow closed-loop systems.
\end{itemize}

The remainder of this paper is organized as follows. 
In Section~\ref{sec:model}, we introduce the physical channel model, starting from the underlying \ac{3D} transport formulation, and derive the \ac{1D} dispersive model under pulsatile flow. 
Section~\ref{sec:channel} derives the corresponding \ac{MC} channel model, first for a straight duct and then for closed-loop channels. 
Section~\ref{sec:results} introduces the considered velocity waveforms and the \ac{3D} \ac{PBS} setup, and validates the analytical \ac{CIR} for both synthetic and physiologically motivated pulsatile flow scenarios. 
Finally, Section~\ref{sec:conclusion} concludes the paper.

\section{Physical Channel Model}\label{sec:model}

In this section, we introduce the physical channel model for the considered closed-loop \ac{MC} system with pulsatile flow. The overall system is illustrated in Fig.~\ref{fig:system_model}, where signaling molecules propagate through a cylindrical channel and may recirculate due to the closed-loop topology.
Starting from the \ac{3D} advection-diffusion equation with spatially and temporally variant flow in cylindrical coordinates, as introduced by Aris~\cite{aris1956}, we obtain a tractable description of the axial transport by deriving a reduced \ac{1D} dispersive transport model for time-variant flow \cite{wang2025pulsatile}. We then adopt a Womersley flow model for the pulsatile velocity field and define the corresponding effective diffusion coefficient.

\subsection{\ac{3D} Transport Model}
In a cylindrical channel, the transport of signaling molecules is governed by \ac{3D} advection and diffusion. Following the classical
formulation by Aris~\cite{aris1956}, the particle concentration $p(x,r,t)$, where
$x$ and $r$ denote the axial and radial coordinates, respectively, and $t$ denotes
time, can be described by the radially symmetric advection-diffusion equation \cite{wang2025pulsatile}
\begin{equation}
\partial_t p(x,r,t)+u(r,t)\partial_x p(x,r,t)
\!=D\!\left(\partial_{xx}p(x,r,t)\!+\!\!\frac{1}{r}\partial_r\!\big(r\,\partial_r p(x,r,t)\big)\right),
\label{eq:3D_adv_diff}
\end{equation}
where $\partial_t$ and $\partial_x$ denote the first order derivatives with respect to time $t$ and axial coordinate $x$, respectively, $\partial_{xx}$ denotes the second order derivative with respect to the axial coordinate $x$, and $\partial_r$ denotes the first order derivative with respect to the radial coordinate $r$. Moreover, $u(r,t)$ is the axial \ac{3D} pulsatile flow velocity profile with pulsation period $T$, and $D$ is the molecular diffusion coefficient. 

\subsection{Dispersive Regime and \ac{1D} Model Reduction}
\label{sec:assumptions}

Following the classical Aris-Taylor dispersion theory and its extension to
time-variant flows, the reduction of \eqref{eq:3D_adv_diff} to a \ac{1D} approximation relies on the following assumptions~\cite{wang2025pulsatile,aris1956}:

\begin{itemize}
  \item The channel is long and slender, i.e., $R \ll L$, where $R$ and $L$
  denote the channel radius and length, respectively.
  
  \item Radial diffusion is much faster than axial transport, such that the
  concentration becomes approximately uniform over the cross section, which requires
\begin{equation}
\frac{R^2}{D} \ll \frac{L}{\bar{u}},
\label{eq:radial_fast}
\end{equation}

 where $\bar{u}$ is the temporal mean of the cross-sectionally averaged axial
flow velocity over one pulsation period $T$, i.e.,
\begin{equation}
\bar{u}
=
\frac{2}{R^{2}T}
\int_{0}^{T}
\int_{0}^{R}
r\,u(r,t)\,\mathrm{d}r\,\mathrm{d}t .
\label{eq:u_bar_def}
\end{equation}

\item Axial dispersion is dominated by shear-induced spreading, which arises from the combined effect of transverse diffusion and the non-uniform flow profile across the channel cross section. In particular, molecular diffusion in the axial direction alone is negligible, namely, 
\begin{equation}
\frac{L}{\bar{u}} \ll \frac{L^2}{D}
\label{eq:axial_dispersion_condition}
\end{equation}
  
\item The flow is laminar, fully developed, and radially symmetric, and the carrier fluid is Newtonian.
\end{itemize}

Under these conditions, solute transport enters the dispersive regime, in which the concentration becomes rapidly homogenized over the cross-section, and the longitudinal dynamics can be described by an effective \ac{1D} advection-diffusion equation with time-variant coefficients. Accordingly, the concentration can be represented by its cross-sectional average as follows
\[
p(x,t) = \frac{2}{R^2}\int_0^R p(x,r,t)\, r\,\mathrm{d}r.
\]
For the case of Womersley velocity profiles, i.e., pulsatile flow profiles in cylindrical channels described by the Womersley solution~\cite{womersley1955}, \eqref{eq:3D_adv_diff} can be transformed into an effective \ac{1D} advection-diffusion equation describing the axial particle transport~\cite[Eq.~(41)]{wang2025pulsatile}
\begin{equation}
\partial_t p(x,t)
+ u(t)\partial_x p(x,t) =
D_{\mathrm{1D}}(t) \partial_{xx} p(x,t).
\label{eq:adv_diff}
\end{equation} 
Here, $u(t)$ is the cross-sectionally averaged axial flow velocity obtained from the underlying \ac{3D} velocity profile $u(r,t)$, and $D_{\mathrm{1D}}(t)$ is the effective time-variant diffusion coefficient that accounts for the combined effect of molecular diffusion and shear-induced dispersion. In the following, we describe the effective time-variant velocity $u(t)$ and diffusion coefficient $D_{\mathrm{1D}}(t)$.

\subsection{Womersley Velocity Profile}
\label{subsec:womersley_flow}

The Womersley velocity profile is commonly used to model pulsatile flow in cylindrical vessels~\cite{womersley1955,11157771,wang2025pulsatile}. Here, we adopt this model for the \ac{3D} axial velocity profile as given in~\cite[Eq.~(31)--(39)]{wang2025pulsatile}. The axial velocity $u(r,t)$ in \eqref{eq:3D_adv_diff} is modeled as a superposition of a steady Poiseuille component and a pulsatile component composed of harmonic Womersley modes, i.e.,
\begin{equation}
u(r,t)=u_{\mathrm{P}}(r)+\Re\!\left\{u_{\mathrm{W}}(r,t)\right\},
\label{eq:u_3d_decomp}
\end{equation}
with the steady contribution $u_{\mathrm{P}}(r)=2\bar{u}\left(1-\frac{r^2}{R^2}\right)$,
where $\bar{u}$ is the cross-sectionally averaged axial velocity~\eqref{eq:u_bar_def}, and $\Re\{\cdot\}$ the real-part operator.

The pulsatile Womersley component in \eqref{eq:u_3d_decomp} is defined as
\begin{equation}
u_{\mathrm{W}}(r,t)
=
\sum_{n=1}^{N}
U_w(\omega_n)\,
\Psi_n(r)\,
e^{\mathrm{j}(n\omega t+\phi_n)},
\label{eq:u_womersley}
\end{equation}
where $U_w(\omega_n)$ denotes the complex amplitude of the $n$-th harmonic with angular frequency $\omega_n = n\omega$, $\phi_n$ is the corresponding phase shift, $N$ is the number of harmonics, and $\mathrm{j}$ is the imaginary unit. The fundamental angular frequency is given by $\omega=2\pi f$, where $f$ is the pulsation frequency.
The radial dependence of each harmonic in \eqref{eq:u_womersley} is described by the Womersley shape function
\begin{equation}
\Psi_n(r)
=
\frac{
J_0\!\left(\mathrm{j}^{3/2}a_n\right)
-
J_0\!\left(\mathrm{j}^{3/2}a_n \frac{r}{R}\right)
}{
J_0\!\left(\mathrm{j}^{3/2}a_n\right)
-
\frac{2J_1\!\left(\mathrm{j}^{3/2}a_n\right)}{\mathrm{j}^{3/2}a_n}
},
\label{eq:womersley_shape}
\end{equation}
where $J_0(\cdot)$ and $J_1(\cdot)$ denote Bessel functions of the first kind of order zero and one, respectively. The corresponding Womersley number of the $n$-th harmonic is defined as
\begin{equation}
a_n = R\sqrt{\frac{n\omega}{\nu}},
\label{eq:womersley_number}
\end{equation}
where $\nu$ in \eqref{eq:womersley_number} is the kinematic viscosity of the fluid. The Womersley number in \eqref{eq:womersley_number} quantifies the relative importance of unsteady inertial effects and viscous diffusion, and indicates how rapidly the velocity profile adapts to temporal variations of the flow~\cite{womersley1955}. Small Womersley numbers correspond to quasi-steady flow conditions where viscous effects dominate and the velocity profile remains close to parabolic, whereas large Womersley numbers indicate inertia-dominated flow with a flattened velocity profile and significant phase lag between the pressure gradient and the flow.

From this \ac{3D} description, the effective \ac{1D} velocity $u(t)$ in \eqref{eq:adv_diff} is obtained by cross-sectional averaging of the \ac{3D} velocity $u(r,t)$ in \eqref{eq:u_3d_decomp} as follows
\begin{equation}
u(t)=\frac{2}{R^2}\int_0^R r\,u(r,t)\,\mathrm{d}r
=\bar{u}\!\left(
1 + \sum_{n=1}^{N} M_n \cos(n\omega t + \phi_n)
\right),
\label{eq:u_t}
\end{equation}
where the dimensionless coefficients $M_n$ are directly related to the amplitudes of the corresponding \ac{3D} Womersley harmonics as
$M_n=\frac{U_w(\omega_n)}{\bar{u}}$.



\subsection{Time-variant Diffusion Coefficient}
Similar to the effective diffusion coefficient for time-invariant systems \cite{Jamali2019ChannelModeling}, the time-variant effective diffusion coefficient in \eqref{eq:adv_diff} is determined by the pulsatile flow velocity described by the Womersley model in Section~\ref{subsec:womersley_flow}. 
For the parameter ranges considered in this paper, the corresponding Womersley numbers $a_n$ are within the regime for which the diffusion coefficient approximation derived in \cite[Eq.~(43)]{wang2025pulsatile} can be directly applied, i.e.,
\begin{equation}
D_{\mathrm{1D}}(t)
=
D + \frac{\bar{u}^2 R^2}{48D}
\left(\frac{u(t)}{\bar{u}}\right)^2.
\label{eq:D1D}
\end{equation}

\section{\ac{MC} Channel Model}\label{sec:channel}
In this section, we derive the \ac{MC} channel model by solving the \ac{1D} advection-diffusion equation in \eqref{eq:adv_diff}. We first consider a straight duct without recirculation and obtain the corresponding solution via a sequence of variable transformations. Subsequently, we extend the model to closed-loop channels.

\subsection{Straight Duct}
\label{subsec:str_duct}

\begin{lemma}[Straight duct solution with pulsatile flow]
\label{lem:straight_duct}
For an impulsive release of particles at $t=0$ and $x=0$, the solution to \eqref{eq:adv_diff} in an infinitely long straight duct can be obtained as follows
\begin{equation}
p_{\mathrm{sd}}(x,t)
=
\frac{1}{\sqrt{2\pi\sigma^2(t)}}
\exp\!\left(
-\frac{[x-\mu(t)]^2}{2\sigma^2(t)}
\right),
\label{eq:gaussian_infinite}
\end{equation}
with time-variant mean $\mu(t)$ and variance $\sigma^2(t)$ given by
\begin{align}
&\mu(t) = \int_0^t u(s)\,\mathrm{d}s, 
&&\sigma^2(t) = 2\int_0^t D_{\mathrm{1D}}(s)\,\mathrm{d}s. \label{eq:mean_var_def}
\end{align}
Inserting the pulsatile velocity model in \eqref{eq:u_t} into \eqref{eq:mean_var_def}, an analytical expression for the mean $\mu(t)$ follows as
\begin{equation}
\mu(t)
=
\bar u\left[
t
+
\sum_{n=1}^{N}
\frac{M_n}{n\omega}
\big(
\sin(n\omega t+\phi_n)-\sin\phi_n
\big)
\right].
\label{eq:mu_closed}
\end{equation}
Moreover, substituting \eqref{eq:u_t} and \eqref{eq:D1D} into \eqref{eq:mean_var_def}
yields an analytical expression for the time-variant variance $\sigma^2(t)$, i.e., 
\begin{equation}
\begin{aligned}
\sigma^2(t)
=
2Dt
+
\frac{R^2\bar u^2}{24D}
\Bigg[
& t
+2\sum_{n=1}^{N} M_n S_n(t)
+\sum_{n=1}^{N} M_n^2 Q_n(t) \\[-0.8em]
& +2\sum_{1\le m<n\le N} M_m M_n P_{m,n}(t)
\Bigg],
\end{aligned}
\label{eq:sigma2_closed}
\end{equation}
with the coefficients 
\begin{align}
S_n(t)
&=
\frac{\sin(n\omega t+\phi_n)-\sin(\phi_n)}{n\omega},
\\[0.6ex]
Q_n(t)
&=
\frac{t}{2}
+
\frac{\sin(2n\omega t+2\phi_n)-\sin(2\phi_n)}{4n\omega},
\\[0.6ex]
P_{m,n}(t)
&=
\frac{1}{2}\Bigg[
\frac{\sin((n-m)\omega t+(\phi_n-\phi_m))
      -\sin(\phi_n-\phi_m)}{(n-m)\omega}
\nonumber\\
&\hspace*{-1cm}\phantom{=.}
+
\frac{\sin((n+m)\omega t+(\phi_n+\phi_m))
      -\sin(\phi_n+\phi_m)}{(n+m)\omega}
\Bigg],\, m<n .
\end{align}
\end{lemma}

\begin{proof}
To obtain solution \eqref{eq:gaussian_infinite} for $p(x,t)$, we transform \eqref{eq:adv_diff} into a diffusion equation with constant coefficients. We first introduce the moving spatial coordinate
\begin{equation}
\xi = x - \mu(t), \qquad \mu(t) = \int_0^t u(s)\,\mathrm{d}s.
\end{equation}
where $\mu(t)$ represents the cumulative axial displacement of particles by the time-variant flow.
Next, we define the transformed concentration
\[
\tilde p(\xi,t) := p(\xi+\mu(t),t),
\]
and since $\frac{\mathrm{d}\mu(t)}{\mathrm{d}t}=u(t)$, the chain rule yields
\begin{align}
\partial_x p(x,t) &= \partial_\xi \tilde p(\xi,t),\\
\partial_{xx} p(x,t) &= \partial_{\xi\xi} \tilde p(\xi,t),\\
\partial_t p(x,t) &= \partial_t \tilde p(\xi,t) - u(t)\partial_\xi \tilde p(\xi,t).
\end{align}
Substituting these expressions into \eqref{eq:adv_diff}, the advection terms cancel and we obtain a diffusion equation with time-variant diffusion coefficient
\begin{equation}
\partial_t \tilde p(\xi,t)=D_{\mathrm{1D}}(t)\partial_{\xi\xi}\tilde p(\xi,t).
\label{eq:transf}
\end{equation}
Next, we introduce a transformed time variable
\begin{equation}
\tau(t)=\int_0^t D_{\mathrm{1D}}(s)\,\mathrm{d}s ,
\end{equation}
which normalizes the time axis such that the diffusion coefficient becomes constant, i.e., $\mathrm{d}\tau/\mathrm{d}t=D_{\mathrm{1D}}(t)$. Let $h(\xi,\tau(t)) = \tilde p(\xi,t)$. Differentiating with respect to $t$ yields
\begin{equation}
\partial_t \tilde p(\xi,t)=D_{\mathrm{1D}}(t)\partial_\tau h(\xi,\tau).
\end{equation}
Substituting into \eqref{eq:transf} leads to a diffusion equation with constant coefficients
\begin{equation}
\partial_\tau h(\xi,\tau)=\partial_{\xi\xi}h(\xi,\tau). \label{eq:trans-diff}
\end{equation}
For an impulsive release of particles at $t=0$ and $x=0$, the transformed initial condition is $h(\xi,0)=\delta(\xi)$. Solving \eqref{eq:trans-diff} in an infinitely long straight duct yields, after returning to the original variables~\cite{Jamali2019ChannelModeling}, the solution $p_\mathrm{sd}(x,t)$ for \eqref{eq:adv_diff} in \eqref{eq:gaussian_infinite}.
\end{proof}

\subsection{Extension to Closed-Loop Channels}

To account for the recirculating topology of closed-loop systems (see Fig.~\ref{fig:system_model}), we extend the solution for the infinitely long straight duct in \eqref{eq:gaussian_infinite} in the following. 


\begin{theorem}[Closed-loop solution with pulsatile flow]
\label{thm:closed_loop}
The concentration $p(x,t)$ in a closed-loop channel of length $L$, i.e., $x\in[0,L]$, with pulsatile flow can be represented by the following wrapped normal distribution
\begin{equation}
p(x,t)
=
\sum_{k=-\infty}^{\infty}
\frac{1}{\sqrt{2\pi\sigma^2(t)}}
\exp\!\left(
-\frac{[x-\mu(t)+kL]^2}{2\sigma^2(t)}
\right),
\label{eq:cir}
\end{equation}
where $k\in\mathbb{Z}$ indexes the contributions associated with different numbers of circulations around the loop. In particular, the term with $k=0$ corresponds to the contribution without a complete circulation, while the terms with $k\neq 0$ capture molecules that have traversed the loop one or multiple times. Equation \eqref{eq:cir} thus extends the straight-duct result in \eqref{eq:gaussian_infinite} to a recirculating channel while preserving the time-variant mean displacement $\mu(t)$ and variance $\sigma^2(t)$ derived in \eqref{eq:mean_var_def}.
\end{theorem}

\begin{proof}
To obtain $p(x,t)$ for closed-loop systems \eqref{eq:cir}, we first restrict the axial coordinate $x$ to the interval $[0,L]$ and impose periodic boundary conditions, i.e.,
\begin{equation}
p(0,t)=p(L,t), \qquad \partial_x p(0,t)=\partial_x p(L,t).
\label{eq:periodic_bc}
\end{equation}
Then, the corresponding closed-loop solution \eqref{eq:cir} can be obtained in two ways. First, one may solve \eqref{eq:adv_diff} directly on $[0,L]$ following the approach described for the straight-duct case in Section \ref{subsec:str_duct}, subject to the periodic boundary conditions in \eqref{eq:periodic_bc}. Second, one may start directly from \eqref{eq:gaussian_infinite} and construct a periodic representation by summing over all integer shifts of length $L$, which effectively wraps the straight-duct solution \eqref{eq:gaussian_infinite} around a circle of circumference $L$, in line with the approach in \cite{vakilipoor2025cam}.
\end{proof}


Finally, to obtain the received signal, we consider a receiver of length $\Delta x_{\mathrm{Rx}}$ located at position $x = x_{\mathrm{Rx}}$ (see Fig.~\ref{fig:system_model}). The observed particle concentration is obtained by integrating the spatial concentration over the receiver region, i.e.,
\begin{equation}
p_{\mathrm{Rx}}(t)
=
\int_{x_{\mathrm{Rx}}-\Delta x_{\mathrm{Rx}}/2}^{x_{\mathrm{Rx}}+\Delta x_{\mathrm{Rx}}/2}
p(x,t)\,\mathrm{d}x.
\label{eq:pRx}
\end{equation}

\section{Simulation Results}\label{sec:results}

In this section, the proposed analytical model in \eqref{eq:cir}, derived for a closed-loop channel, is validated by comparison to results obtained from \ac{3D} \ac{PBS}. First, we introduce the considered flow-velocity waveforms based on \eqref{eq:u_t}. Subsequently, we present the simulation setup used for the \ac{PBS}. Finally, we analyze the resulting received signals for different system parameters and compare the analytical model with both \ac{PBS} results and the time-invariant model in \cite{vakilipoor2025cam}.

\subsection{Velocity Waveforms}
\label{sec:Vel_waveforms}

For validation, we consider both synthetic and physiologically motivated pulsatile flow velocity waveforms in order to assess the model under controlled as well as realistic flow conditions. 
For all considered waveforms, the \ac{1D} velocity is described by \eqref{eq:u_t}. The corresponding \ac{3D} flow field used in the \ac{PBS} is constructed according to the Womersley model in Section~\ref{subsec:womersley_flow}, ensuring that the \ac{3D} velocity field has the same temporal mean and harmonic content.

\subsubsection{Synthetic Velocity Waveforms}
For the synthetic study, we consider two periodic waveforms. The first is a sinusoidal waveform given by
\begin{equation}
u_{\mathrm{sin}}(t)=\bar{u}\left(1+A\sin(2\pi f t)\right),
\label{eq:u_sin_sim}
\end{equation}
where $A=0.5$ denotes the oscillation amplitude and $f$ is the pulsation frequency. This corresponds to a single-harmonic representation of \eqref{eq:u_t} with $N=1$, $M_1=A$, and $\phi_1=-\pi/2$.

The second synthetic waveform is a pulsed waveform with duty cycle $d=0.2$ and temporal mean $\bar{u}$, defined as
\begin{equation}
u_{\mathrm{pulse}}(t)=
\begin{cases}
\frac{\bar{u}}{d}, & 0 \le \mathrm{mod}(t,T) < dT,\\
0, & dT \le \mathrm{mod}(t,T) < T,
\end{cases}
\label{eq:u_pulse_sim}
\end{equation}
where $T=1/f$ is the pulsation period, and $\mathrm{mod}(t,T)$ denotes the remainder of $t$ divided by $T$. Since this ideal rectangular waveform cannot be directly represented by the finite harmonic expansion in \eqref{eq:u_t}, it is approximated by a truncated Fourier series. In particular, the waveform is expressed in the form of \eqref{eq:u_t} by retaining the first $N$ harmonics of its Fourier series representation, where $N$ is chosen sufficiently large to adequately capture the shape of the pulse. The corresponding coefficients are obtained from the Fourier series of the rectangular waveform and are given by
\begin{equation}
M_n = \sqrt{A_n^2 + B_n^2}, \qquad
\phi_n = \operatorname{atan2}(-B_n, A_n),
\end{equation}
with
\begin{equation}
A_n = \frac{\sin(2\pi n d)}{\pi n d}, \qquad
B_n = \frac{1-\cos(2\pi n d)}{\pi n d},
\end{equation}
for $n=1,\dots,N$. Hence, the pulsed waveform used in the simulations is a finite-harmonic approximation of an ideal rectangular pulse.
The synthetic sinusoidal \eqref{eq:u_sin_sim} and pulsed \eqref{eq:u_pulse_sim} velocity waveforms are illustrated in Fig.~\ref{fig:waveform_sine} and Fig.~\ref{fig:waveform_pulse}, respectively, for $\bar{u}=1\cdot10^{-4}\,\si{\meter\per\second}$ and $f=0.5\,\si{\hertz}$.

\subsubsection{Realistic Velocity Waveform}
In addition to the synthetic waveforms, we consider a physiologically motivated pulsatile waveform adopted from \cite{11157771}. Therefore, we approximate the waveform from \cite{11157771} by fitting the parameters of \eqref{eq:u_t} using $N=12$ harmonics. The resulting fit, denoted by $u_{\mathrm{phys}}(t)$, has fundamental frequency $f = 1.15\,\si{\hertz}$, amplitude coefficients
\[
\begin{aligned}
(M_1,\ldots,M_{12}) = (&0.548,\,0.684,\,0.373,\,0.489,\,0.352,\,0.166,\\
&0.253,\,0.135,\,0.195,\,0.134,\,0.162,\,0.190),
\end{aligned}
\]
and phase shifts
\[
\begin{aligned}
(\phi_1,\ldots,\phi_{12}) = (&-0.869,\,-1.826,\,-3.009,\,3.137,\,1.815,\,1.944,\\
&1.252,\,0.727,\,0.287,\,-0.504,\,-0.605,\,-1.307).
\end{aligned}
\]
The fitted waveform is then normalized with respect to its temporal mean and thus preserves its shape independently of the mean velocity $\bar{u}$, which is selected according to the considered simulation scenario. The resulting fit of the physiologically motivated waveform obtained using \eqref{eq:u_t} is shown in Fig.~\ref{fig:waveform_real} for $\bar{u}=1\cdot10^{-4}\,\si{\meter\per\second}$.

Overall, the velocity model in \eqref{eq:u_t} provides a flexible framework that can capture a wide range of velocity waveforms, including both physiologically realistic pulsatile flows in in-body environments and synthetic waveforms used in controlled experimental or microfluidic setups.

\begin{figure}[t]
  \centering
 \begin{subfigure}[t]{0.95\columnwidth}
   \centering
   \includegraphics[width=\linewidth]{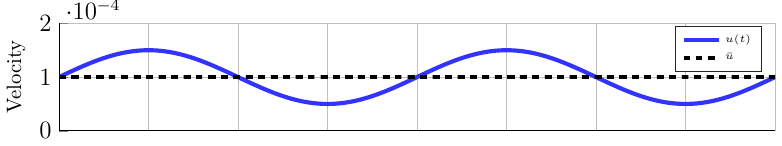}
   \vspace*{-0.6cm}
   \caption{Sinusoidal waveform $u_{\mathrm{sin}}(t)$}
   \label{fig:waveform_sine}
 \end{subfigure}\par\vspace{0.3em}

 \begin{subfigure}[t]{0.95\columnwidth}
   \centering
   \includegraphics[width=\linewidth]{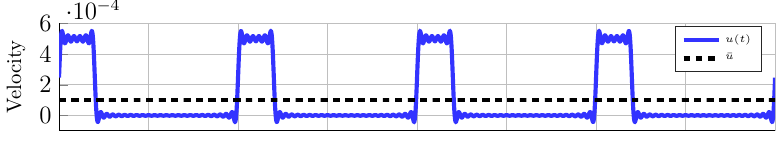}
   \vspace*{-0.6cm}
   \caption{Pulsed waveform $u_{\mathrm{pulse}}(t)$}
   \label{fig:waveform_pulse}
 \end{subfigure}\par\vspace{0.3em}%

 \begin{subfigure}[t]{0.95\columnwidth}
   \centering
  \includegraphics[width=\linewidth]{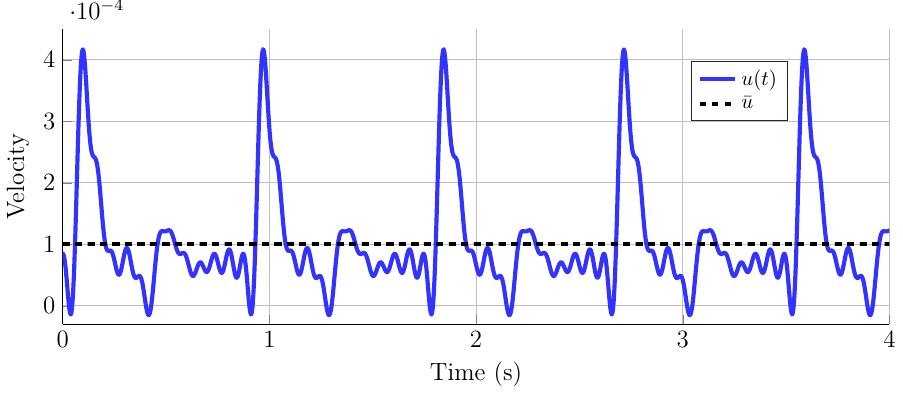}
  \vspace*{-0.6cm}
   \caption{Physiologically motivated waveform $u_{\mathrm{phys}}(t)$}
   \label{fig:waveform_real}
 \end{subfigure}
\vspace*{-0.15cm}
  \caption{\textbf{Velocity waveforms} used for validating the proposed analytical model in \eqref{eq:cir}: synthetic sinusoidal (top) and pulsed (center) waveforms, and a physiologically motivated waveform (bottom), shown for $\bar{u}=1\cdot10^{-4}\,\si{\meter\per\second}$.}
  \label{fig:waveforms_all}
\end{figure}

\subsection{Particle-Based Simulation Setup}

We implement the \ac{PBS} in a \ac{3D} cylindrical channel of radius $R$ and length $L$. The domain is defined for $x \in [0,L]$, and the closed-loop topology is realized by applying periodic boundary conditions in the axial direction, i.e., particles that move beyond $x=L$ (or $x<0$) are reinserted at $\mathrm{mod}(x,L)$. Reflecting boundary conditions are imposed at the cylindrical walls.

At time $t=0$, $N_{\mathrm{p}} = 5\cdot10^5$ particles, which are released at $x=0$, are uniformly distributed over the circular cross section.
The subsequent particle propagation is simulated in discrete time steps with step size $\Delta t = 10^{-4}\,\mathrm{s}$ over a total duration of $T = 20\,\mathrm{s}$. In each time step, particle positions are updated by combining deterministic advection and stochastic diffusion. The axial displacement is described by the \ac{3D} flow velocity $u(r,t)$ in \eqref{eq:u_3d_decomp}. Diffusion is modeled by adding independent Gaussian increments with variance $2D\Delta t$ in all spatial directions.

The received signal is obtained by counting the number of particles within a receiver region of length $\Delta x_{\mathrm{Rx}}$ centered at position $x_{\mathrm{Rx}}$. This region corresponds to a cylindrical slice of volume $V_{\mathrm{Rx}} = \pi R^2 \Delta x_{\mathrm{Rx}}$. The particle concentration is estimated by normalizing the number of particles inside the receiver volume by the total number of released particles.

The parameters were chosen according to~\cite{Jamali2019ChannelModeling,vakilipoor2025cam} such that the system operates in the dispersive regime and the \ac{1D} model assumptions are satisfied. Unless stated otherwise, the parameters are fixed to
$D = 5\cdot10^{-9}\,\si{\square\meter\per\second}$,
$R = 50\,\si{\micro\meter}$,
$L = 1\,\si{\milli\meter}$\footnote{Although physiological vascular loop lengths are typically much larger (cf.~\cite{vakilipoor2025cam}), we adopt a reduced loop length of $L=1\,\si{\milli\meter}$ in our simulations to enable clear visualization of repeated circulations within a reasonable simulation time and facilitates comparison with the 3D PBS results.},
$x_{\mathrm{Rx}} = 0.3\,\si{\milli\meter}$, and
$\Delta x_{\mathrm{Rx}} = 0.1\,\si{\milli\meter}$.
Moreover, we assume a fluid density of $\rho = 1060\,\si{\kilogram\per\meter\cubed}$
and a dynamic viscosity of $\mu = 3\times10^{-3}\,\si{\pascal\second}$,
corresponding to a kinematic viscosity
$\nu = \mu/\rho \approx 2.83\times10^{-6}\,\si{\meter\squared\per\second}$. For these parameters, the resulting Womersley numbers $a_n = R\sqrt{n\omega/\nu}$ remain below unity for all considered harmonics, indicating operation in the low-Womersley-number regime and justifying the use of the analytical model in \eqref{eq:D1D}.

\subsection{Results for Synthetic Velocity Waveforms}

\begin{figure}[t]
  \centering
 \begin{subfigure}[t]{0.48\textwidth}
   \centering
  \includegraphics[width=\linewidth]{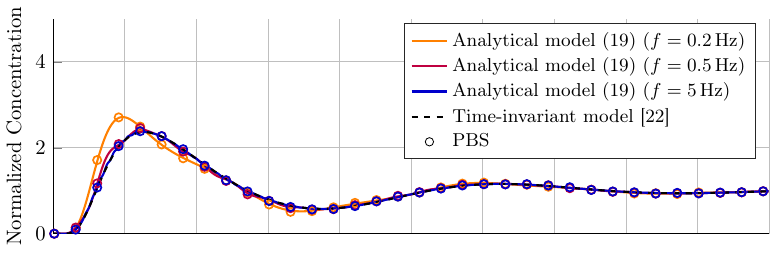}
  \vspace*{-0.6cm}
   \caption{Sinusoidal waveform}
   \label{fig:sine_f}
 \end{subfigure}\par\vspace{0.3em}%
 \begin{subfigure}[t]{0.48\textwidth}
   \centering
   \includegraphics[width=\linewidth]{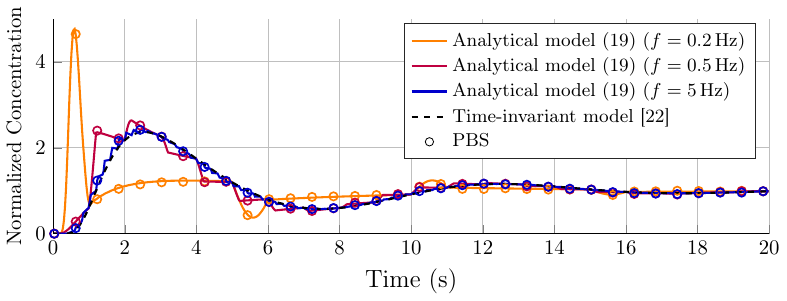}
   \vspace*{-0.6cm}
   \caption{Pulsed waveform}
   \label{fig:pulse_f}
 \end{subfigure}
\vspace*{-0.5cm}
\caption{Normalized received signal for different pulsation frequencies $f$, for sinusoidal (top) and pulsed (bottom) velocity waveforms}
\label{fig:synthetic_f}
\end{figure}

First, we investigate the synthetic velocity waveforms introduced in Section~\ref{sec:Vel_waveforms} (cf. Fig.~\ref{fig:waveform_sine} and Fig.~\ref{fig:waveform_pulse}), where either the pulsation frequency $f$ or the mean velocity $\bar{u}$ is varied.
For better comparability, all results are normalized by the equilibrium concentration
$p_\infty=\Delta x_{\mathrm{Rx}}/L$, obtained from $\lim_{t\to\infty} p(x,t)=1/L$.

Figure~\ref{fig:synthetic_f} shows the normalized received signal for different pulsation frequencies $f$ for both synthetic velocity waveforms, for $\bar{u}=1\cdot10^{-4}\,\si{\meter\per\second}$. 
Since the time-invariant model in \cite{vakilipoor2025cam} assumes a constant flow velocity, it does not depend on the pulsation frequency. Therefore, it yields a single reference curve, which is identical for all values of $f$. First, we observe that the results obtained from the proposed analytical model in \eqref{eq:cir} (solid curves) are in excellent agreement with the \ac{PBS} results (circle markers) for all considered frequencies and both waveforms. 
Furthermore, it can be observed that, for increasing $f$, the received signal obtained from the time-variant model in \eqref{eq:cir} converges to that of the corresponding steady-flow system with constant velocity $\bar{u}$ (dashed curve). This behavior can be explained by a temporal averaging effect: at higher frequencies, the pulsation period becomes small compared to the particle propagation time, such that particles experience multiple velocity oscillations during their propagation and effectively perceive the mean flow velocity.

\begin{figure}[t]
  \centering
 \begin{subfigure}[t]{0.48\textwidth}
   \centering
   \includegraphics[width=\linewidth]{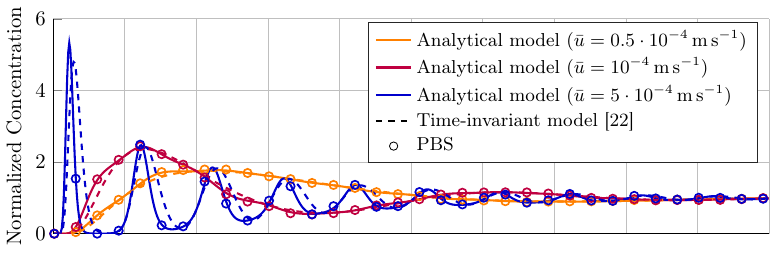}
   \vspace*{-0.6cm}
   \caption{Sinusoidal waveform}
   \label{fig:sine_U}
 \end{subfigure}\par\vspace{0.3em}%
 \begin{subfigure}[t]{0.48\textwidth}
   \centering
   \includegraphics[width=\linewidth]{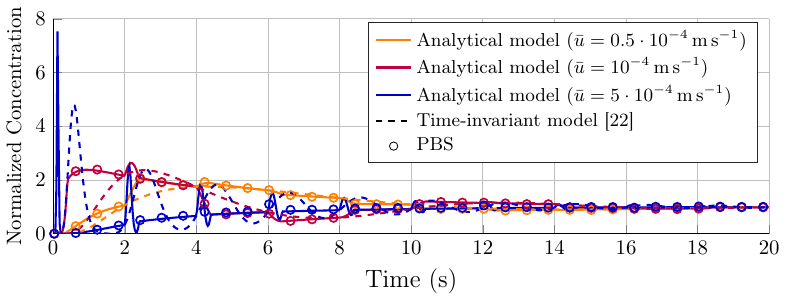}
   \vspace*{-0.6cm}
   \caption{Pulsed waveform}
   \label{fig:pulse_U}
 \end{subfigure}
\vspace*{-0.5cm}
\caption{Normalized received signal for different mean velocities $\bar{u}$, for sinusoidal (top) and pulsed (bottom) velocity waveforms}  \label{fig:synthetic_U}
\end{figure}

Figure~\ref{fig:synthetic_U} shows the normalized received signal for different values of the mean velocity $\bar{u}$ for both synthetic waveforms, for fixed $f=0.5\,\si{\hertz}$. Similar to the results in Fig.~\ref{fig:synthetic_f}, the results obtained from the proposed analytical model in \eqref{eq:cir} (solid curves) show excellent agreement with the \ac{PBS} results (circle markers) for all considered $\bar{u}$ values.
Moreover, as the temporal mean velocity $\bar{u}$ increases, the received signal peaks occur earlier and become sharper, which is consistent with faster advective transport. In addition, more pronounced deviations from the corresponding steady-flow response can be observed, indicating a stronger influence of the pulsatile flow. This behavior can be attributed to the fact that, for larger $\bar{u}$, particle transport becomes increasingly dominated by advection, such that particles more closely follow the instantaneous velocity variations and are therefore more strongly influenced by the pulsations.
In contrast, for decreasing $\bar{u}$, the received signal becomes increasingly similar to that of the corresponding steady-flow system. This is because diffusion plays a more dominant role in this regime, effectively smoothing out the temporal variations of the flow and reducing the impact of pulsations on the particle trajectories.

\subsection{Results for Realistic Velocity Waveform}

Next, we consider the physiologically motivated velocity waveform introduced in Section~\ref{sec:Vel_waveforms} (cf. Fig.~\ref{fig:waveform_real}), and vary the mean velocity $\bar{u}$, the receiver position $x_{\mathrm{Rx}}$, and the diffusion coefficient $D$.
As in the previous experiments, all results are normalized by the equilibrium concentration $p_\infty$.

Figure~\ref{fig:pbs} shows the results from the proposed model in \eqref{eq:cir}, from \ac{PBS} (circle markers), and from the time-invariant model in \cite{vakilipoor2025cam} (dashed curves), for different receiver positions $x_{\mathrm{Rx}}$ with fixed mean flow velocity $\bar{u}=2\cdot10^{-4}\,\si{\meter\per\second}$ (top plot), for different temporal mean velocities $\bar{u}$ with fixed receiver position $x_{\mathrm{Rx}}=0.3\,\mathrm{mm}$ (center plot), and for different diffusion coefficients $D$ with fixed $x_{\mathrm{Rx}}=0.3\,\mathrm{mm}$ and $\bar{u}=2\cdot10^{-4}\,\si{\meter\per\second}$ (bottom plot). First, we observe that the proposed model again is in excellent agreement with the results from \ac{PBS} in all considered scenarios.
Moreover, the time-invariant model from \cite{vakilipoor2025cam} is not able to capture the pulsations, while it still approximates the overall shape of the received signal. The influence of the pulsatile flow becomes less pronounced for increasing receiver distance $x_\mathrm{Rx}$, decreasing temporal mean velocity $\bar{u}$, increasing diffusion coefficient $D$, and over time. These observations are consistent with those obtained for the synthetic velocity waveforms and can be explained by the fact that, in all these cases, the role of advective transport is reduced relative to diffusion, leading to more dispersive system dynamics~\cite{Jamali2019ChannelModeling}. 
Finally, for increasing $D$, the received signal becomes more spread in time and the amplitude of the first peak decreases. More importantly, the influence of the pulsatile flow becomes less pronounced, as reflected by the reduced deviation between the time-variant and the corresponding time-invariant model in Fig.~\ref{fig:pbs_D}. This behavior can be attributed to the stronger diffusive spreading, which smooths out the temporal variations of the flow and reduces the impact of pulsations on the particle trajectories.

Overall, the proposed model provides a versatile analytical framework that enables the characterization of pulsatile flow effects for both synthetic and realistic velocity waveforms, while also identifying regimes in which pulsations become negligible and simpler steady-flow models can be employed instead.

\begin{figure}[t]
  \centering
    
 \begin{subfigure}[t]{0.99\columnwidth}
   \centering
   \includegraphics[width=\linewidth]{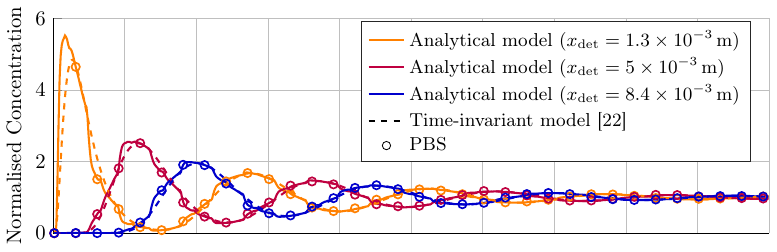}
   \vspace*{-0.6cm}
   \caption{Different $x_{\mathrm{Rx}}$ for fixed $\bar{u}=2\cdot10^{-4}\,\si{\meter\per\second}$.}
   \label{fig:pbs_xdet}
 \end{subfigure}\par\vspace{0.3em}
 \begin{subfigure}[t]{0.99\columnwidth}
   \centering
   \includegraphics[width=\linewidth]{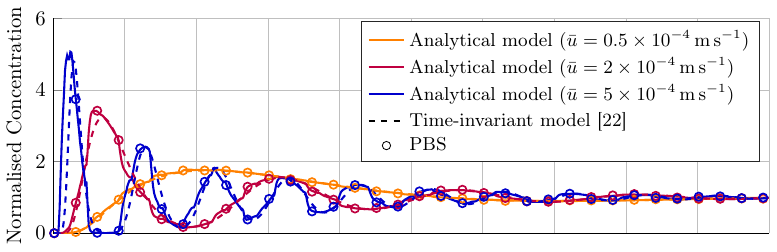}
   \vspace*{-0.6cm}
   \caption{Different $\bar{u}$ for fixed $x_{\mathrm{Rx}}=0.3\,\mathrm{mm}$.}
   \label{fig:pbs_ubar}
 \end{subfigure}\par\vspace{0.3em}
 \begin{subfigure}[t]{0.99\columnwidth}
   \centering
   \includegraphics[width=\linewidth]{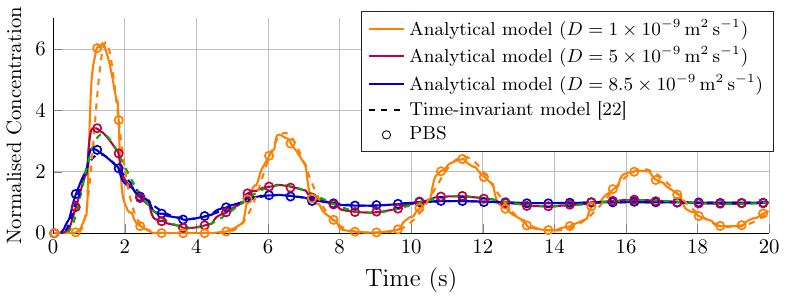}
   \vspace*{-0.6cm}
   \caption{Different $D$ for fixed $x_{\mathrm{Rx}}=0.3\,\mathrm{mm}$ and $\bar{u}=2\cdot10^{-4}\,\si{\meter\per\second}$.}
   \label{fig:pbs_D}
 \end{subfigure}
\vspace*{-0.1cm}
  \caption{\textbf{Physiologically motivated pulsatile flow:} Normalized received signals for different receiver positions $x_{\mathrm{Rx}}$, temporal mean velocities $\bar{u}$, and diffusion coefficients $D$.}
 \Description{Three vertically stacked plots showing normalized received concentration versus time for different receiver positions, mean flow velocities, and diffusion coefficients, comparing the proposed analytical model, particle-based simulations, and the time-invariant model.}
  \label{fig:pbs}
\end{figure}

\section{Conclusion}\label{sec:conclusion}

In this paper, an analytical channel model for dispersive closed-loop \ac{MC} systems with pulsatile flow was proposed. The model incorporates a time-variant flow velocity and yields an analytical expression for the \ac{CIR} in the form of a wrapped Normal distribution with time-variant mean and variance. The model was validated by \ac{3D} \acp{PBS} for both synthetic periodic and a physiologically motivated pulsatile velocity waveform, showing excellent agreement. The results demonstrate that pulsatile flow can noticeably influence the temporal shape of the received signal compared to the corresponding steady-flow case. Future work includes the incorporation of realistic physiological data into the proposed framework, a comparison with experimental measurements, and the study of how spatial variations of the flow field and the attenuation of pulsatility along the propagation path influence signal propagation.

\bibliographystyle{ACM-Reference-Format}
\bibliography{references}

\end{document}